\documentclass{llncs}
\usepackage{blindtext}
\usepackage{thm-restate}
\usepackage{amsmath,amsfonts,graphicx}
\usepackage{mathrsfs}
\usepackage{enumitem}
\usepackage{tikz}
\usetikzlibrary{automata,positioning}
\usepackage{graphicx}
\usepackage{algorithm}
\usepackage[noend]{algpseudocode}
\usepackage{mathtools}
\usepackage[capitalise]{cleveref}
\usepackage{xcolor}
\definecolor{petrol}{RGB}{117,181,190}
\usepackage{mdframed}
\usepackage{caption,subcaption}

\if{}

\newtheorem{lemma}{Lemma}
\newtheorem{remark}{Remark}

\fi 

\newcommand{\R}{\mathbb{R}_{>0}}

\newcommand{\C}{{\cal C}} % actual chain
\newcommand{\CS}{\Lambda} %chain selection rule
\newcommand{\hon}{\mathcal{H}} %Honest Party
\newcommand{\hsp}{s^\hon} %honest space
\newcommand{\hsptil}{{\Tilde{s}}^h} %honest space tilde
\newcommand{\asp}{s^\adv} %adversarial space
\newcommand{\asptil}{{\Tilde{s}}^a} %adversarial space prime
\newcommand{\spratio}{\phi} %honest to adversarial space ratio
\newcommand{\dynratio}{1+\varepsilon} %rate of change in dynamic availability  
\newcommand{\replottime}{\rho} %alternate rate of change in dynamic availability
\newcommand{\sprof}[1]{\mathcal{S}(#1)}
\newcommand{\tent}{\Gamma}

\newcommand{\adv}{{\cal A}}
\newcommand{\la}{\hookleftarrow}
\newcommand{\lint}[1]{\left\lceil #1 \right\rceil} %least int greater than x
\newcommand{\eqdef}{\stackrel{\tiny \sf def}{=}}
\newcommand{\til}{\widetilde}

\newcommand{\ahad}[1]{}
\newcommand{\kcomm}[1]{}

\newcommand{\hla}{\hookleftarrow}

\newcommand\blfootnote[1]{%
  \begingroup
  \renewcommand\thefootnote{}\footnote{#1}%
  \addtocounter{footnote}{-1}%
  \endgroup
}

\title{On the (in)security of Proofs-of-Space based Longest-Chain Blockchains}
 \author{
 	Mirza Ahad Baig$^1$
     \and Krzysztof Pietrzak$^1$
	}
\institute{
 $^1$ISTA, Austria
}

\begin{document}

\maketitle

\begin{abstract}
The Nakamoto consensus protocol underlying the Bitcoin blockchain uses proof of work as a voting mechanism. Honest miners who contribute hashing power towards securing the chain try to extend the longest chain they are aware of. Despite its simplicity, Nakamoto consensus achieves meaningful security guarantees assuming that at any point in time, a majority of the hashing power is controlled by honest parties. This also holds under ``resource variability'', i.e., if the total hashing power varies greatly over time.

Proofs of space (PoSpace) have been suggested as a more sustainable replacement for proofs of work. Unfortunately, no construction of a ``longest-chain'' blockchain based on PoSpace, that is secure under dynamic availability, is known. In this work, we prove that without additional assumptions no such protocol exists. We exactly quantify this impossibility result by proving a bound on the length of the fork required for double spending as a function of the adversarial capabilities. This bound holds for any chain selection rule, and we also show a chain selection rule (albeit a very strange one) that almost matches this bound.  

Concretely, we consider a security game in which the honest parties at any point control $\phi>1$ times more space than the adversary. The adversary can change the honest space by a factor $1\pm \varepsilon$ with every block (dynamic availability), and ``replotting'' the space (which allows answering two challenges using the same space) takes as much time as $\rho$ blocks.

We prove that no matter what chain selection rule is used, in this game the adversary can create a fork of length $\phi^2\cdot \rho / \varepsilon$ that will be picked as the winner by the chain selection rule.

We also provide an upper bound that matches the lower bound up to a factor $\phi$. There exists a chain selection rule (albeit a very strange one) which in the above game requires forks of length at least $\phi\cdot \rho / \varepsilon$.

Our results show the necessity of additional assumptions to create a secure PoSpace based longest-chain blockchain. The Chia network in addition to PoSpace uses a verifiable delay function. 
Our bounds show that an additional primitive like that is necessary.
%and we observe that this indeed gives a secure chain under dynamic availability in the idealized model we use for our lower bounds.

%Concretely, we consider a chain that proceeds in slots, where with each slot a fresh challenge is sampled. An adversary decides on the amount of space controlled by honest parties but can change it by at most an $\gamma$ fraction per slot. The adversary can ``replot" the space (and thus create two proofs using the same space) in time $T$ slots, and at any time point, the adversary controls a $1/s$ fraction of the space available to the honest parties. No matter what chain selection rule is used, in this setting, the adversary can always create a fork (and thus double spend) by creating a fork that goes back at most $s^2\cdot T/\gamma$ slots.

%Our new chain selection rule matching this bound takes two chains $A$ and $B$, and picks $A$ iff it's harder to ``fake'' chain $A$ assuming $B$ is the honest chain (using replotting and boosting) than the other way round. This rule is not transitive, we observe that also the 

\end{abstract}

\blfootnote{This research was funded in whole or in part by the Austrian Science Fund (FWF) 10.55776/F85}

\section{Introduction}
Bitcoin was the first successful digital currency. What set it apart from previous attempts like Digicash~\cite{digicash} was the fact that it is permissionless. 
This means it is decentralized -- so no single entity can shut it down or censor transactions -- and moreover, everyone can participate in maintaining and securing the currency.

The key innovation in Bitcoin is the blockchain which realizes a ``decentralized ledger''. In the case of a digital currency, this ledger simply records all the transactions, but it can also hold richer data like smart contracts~\cite{ethereum}.

A blockchain is a hash-chain 
%$$b_i^j\eqdef 
$b_0\hookleftarrow b_1 \hookleftarrow b_2 \ldots \hookleftarrow b_j$ where each $b_i$ is a data block that contains a ``payload'' (transactions, a time stamp, etc.), a hash $h_i=H(b_{i-1})$ of the previous block, and a ``proof of work'' (PoW) $\pi_i$.

The collision-resistance of the hash function $H$ ensures that a block $b_i$ commits all the previous blocks. The main novelty is the way proofs of work are used to make it computationally costly to add a block: To create a valid block $b_i$ one must find a value $\pi_i$ such that the hash of the previous block and $\pi_i$ is below some threshold
$$
0.H(b_{i-1},\pi_i) < 1/D 
$$
here we think of the hash as a binary string: if the difficulty $D$ is $2^k$, then the hash $H(b_{i-1},\pi_i)$ must start with at least $k$ $0$'s. Parties called miners compete to find PoWs to extend the latest block. They are incentivized by rewards (block rewards and transaction fees) to contribute computing power towards this task. Bitcoin is permissionless in the sense that everyone can be a miner and the protocol does not need to know who currently participates~\cite{lewispye2024permissionlessconsensus}. Bitcoin can be shown to be secure (in particular, it does not allow for double spending), assuming that a majority of the hashing power is controlled by honest parties who follow the protocol rules. The most important rule just states that a miner should always work towards extending the heaviest valid chain (typically, the heaviest chain is also the longest one, hence the name ``longest chain'') they are aware of. Blockchains following this general rule are called ``longest-chain blockchains'', the protocol itself is referred to as ``Nakamoto consensus''.

\paragraph{Alternative Proof Systems.} 
Nakamoto consensus uses computation as a resource so that a miner who holds an $\alpha$ fraction of the total resource will contribute an $\alpha$ fraction of all blocks in expectation, and thus get roughly an $\alpha$ fraction of the rewards. 

Using computation as a resource has several negative implications. The main one is the ecological impact: currently Bitcoin mining is burning roughly as much energy as the Netherlands. It is thus natural to look for a more ``sustainable'' resource that could replace hashing power in a longest-chain blockchain.

The most investigated alternative are \emph{proofs of stake} (PoStake), where the coins as recorded on the blockchain serve as a resource. More precisely, miners can stake their coins, which takes them temporarily out of circulation. 
They can then participate in the mining process, getting a fraction of the rewards which is proportional to the fraction of their stakes coins.

PoStake is extremely appealing as it is basically wasteless as it is not a ``physical'' resource, but it raises many technical questions and conceptual issues. One argument that is often raised is that PoStake is not really permissionless as the only way to participate in mining lies in acquiring coins from a limited supply in the first place.

In~\cite{DFKK2015} \emph{proofs of space} (PoSpace) were introduced. A PoSpace is a proof system where a prover convinces a verifier that it ``wastes'' disk space. The motivation for this notion was a replacement for proofs of work which is still a ``physical resource '', and thus does not share many of the shortcomings of PoStake, but is also much more sustainable than proofs of work. 

\paragraph{Proofs of Space.}A proof of space~\cite{DFKK2015} is an interactive protocol between a prover $P$ and a verifier $V$. The main protocol parameter is a value $N$ determining the disk space of an (honest) prover (a typical value would be $N=2^{43}$ bits, which corresponds to one TB). In an initialization phase, which is executed once, the honest prover initializes his disk space space of size $N$ with a file $S$, called a ``plot''. This phase should be very efficient for the verifier (or not involve the verifier at all~\cite{AbusalahACKPR17}), while the prover should run in time $\tilde O(N)$. This is basically optimal as they must run in time $N$ to just ``touch'' the entire disk space.

After the initialization phase, the prover can create valid proofs for random challenges very efficiently, in particular, only accessing a tiny portion of its local file $S$. The security property of a PoSpace states that any prover who instead of $S$ stores some data $S'$ of some size that is ``sufficiently'' smaller than $S$, will fail to ``efficiently''  create a proof for a random challenge with ``significant'' probability. 

We will not discuss what exactly ``sufficiently'' and ``significant'' means here. Let us mention that for the application to longest-chain blockchains it is sufficient that for any $0<\alpha< 1$, a prover storing $\alpha\cdot N$ bits will fail on a $1-\alpha$ fraction of the challenges.

Concerning the ``efficiently'' in the statement above, note that a malicious prover can always create a valid proof even when storing almost nothing by simply running the initialization procedure after getting the challenge to create the plot $S$, and then computing the proof using the honest algorithm. Thus the best we can hope for is that a malicious prover needs $\tilde\Omega(N)$ computational work (i.e., the cost of computing the plot) when only storing a sufficiently compressed plot $S'$. 

The observation above also implies that a prover with $N$ space can ``pretend'' to have $k\cdot N$ space by creating $k$ different plots sequentially using $k\cdot \tilde O(N)$ work. Note that when attacking a blockchain, one would need to do the replotting afresh for every block, as the challenge for a block is only known once the previous block is computed. Thus, creating a proof using such a \emph{replotting attack} is extremely expensive compared to creating proofs honestly, and this attack is presumably not an issue when blocks arrive sufficiently frequently.
%, like it is the case in the blockchain setting, where we can require miners to compute a proof for every block. 
\emph{The results of this paper show that this intuition is wrong.}

\paragraph{Longest-Chain Blockchains from Proofs of Space.}
To construct a longest-chain blockchain from PoSpace we can use Nakamoto consensus, but replace the PoW with PoSpace. There are various challenges one must address which we outline below.
\begin{description}
\item[Interactive Resource Initialization:]
In Bitcoin, a miner with some mining hardware can start participating in mining at any time. For PoSpace this is in general not the case as there's an initialization phase. The earliest PoSpace longest-chain proposal (which remained purely academic) is Spacemint~\cite{spacemint}, which uses the pebbling-based PoSpace from~\cite{DFKK2015}. This PoSpace has an interactive initialization phase after which the verifier holds a type of commitment to the plot created by the prover. In Spacemint the chain plays the role of the verifier, and the commitment must be uploaded by a miner to the chain as a special transaction before they can start mining. The function-inversion-based PoSpace from~\cite{AbusalahACKPR17} has a non-interactive initialization, i.e., the verifier is not involved at all, and thus it can be used like a PoW in Bitcoin. This PoSpace is used in the Chia network~\cite{chia} blockchain.

\item[Bock-Arrival Times:]
In Bitcoin, the arrival time of blocks can be controlled by setting the difficulty. Unlike PoW, PoSpace (also PoStake) are efficient proof systems, where once the resource (a plot or staked coins) is available, creating a proof is cheap and fast, so we need another mechanism. The easiest approach is to simply assume all parties have clocks and specify that blocks are supposed to arrive in specific time intervals, say once every minute. This is the approach taken by Spacemint~\cite{spacemint} or Ouroboros~\cite{ouroboros}, while in the Chia network~\cite{chia}  verifiable-delay functions (VDFs) are used to enforce some clock-time between the creation of blocks.
\item[Costless Simulation/Grinding:]
The key difference between PoW and ``efficient'' proof systems like PoSpace or PoStake is the fact that producing proofs for $k>1$ different challenges require $k$ times as much of the resource in PoW, but it makes hardly a difference for PoSpace or Postake, as producing a proof is extremely cheap compared to acquiring the resource in the first place. This ``costless simulation'' property creates various issues in the blockchain setting. 
One such issue is grinding attacks. Consider a setting where an adversary can influence the challenge, a typical example is a blockchain like Bitcoin where the challenge for the next block depends on the current block, and the miner that creates the current block can e.g. choose which transactions to add. Such an attacker can ``grind'' through many different blocks until they find one that gives a challenge they like (say because with this challenge they can also win the next block).

A canonical countermeasure against grinding first proposed in Spacemint~\cite{spacemint} is to ``split'' the chain in two. One chain only holds canonical values like proofs and is used for creating challenges, while another chain holds all the ``grindable'' values (transactions, time-stamps, etc.).
\item[Costless Simulation/Double-Dipping:]
Even once grinding is no longer an issue, with costless simulation an adversary can still cheaply try to extend many of the past blocks, this way growing a tree rather than a chain. This strategy has been proven to virtually increase the adversarial resource by a factor of $e=2.72$~\cite{chia}. An elegant countermeasure against this attack was proposed in \cite{Bagariaetal2022-proof-of-stake}, the basic idea is to only use the $k$th block for computing challenges, this ``correlated randomness'' technique decreases the advantage as $k$ increases.
\item[Costless Simulation/Bootstrapping:]
An adversary having some resource (space or stake) $N$ can immediately create a long chain. Typically one would make up the time-stamps for this chain so it looks like a legit chain that has been created over a long period. In context of PoStake this is a well-known problem, while Spacemint~\cite{spacemint} was the first instance this appeared in literature for PoSpace.
\item[Replotting:]
In PoSpace a prover who has space of size $N$ and gets a challenge $c$ can pretend to have much more space by re-initializing the same space $k$ times with different identifiers, this way creating $k$ proofs pretending to have $k\cdot N$ bits of space. 
%This ``attack'' is inherent to the notion of PoSpace. 
As replotting is fairly expensive, in a context like Blockchains, where challenges arrive frequently, it seems impossible (or at least not rational) to continuously replot. In this paper we show that this intuition is flawed; replotting attacks, in combination with bootstrapping, are used in our lower bound showing no PoSpace longest-chain blockchain is secure under resource  variability. This was identified as a problem in~\cite{spacemint}. 

\end{description}

\subsection{Resource Variability}
In this work, we prove that no PoSpace based longest-chain blockchain can be secure. 
%under \emph{resource variability}\footnote{In an earlier version of this paper this property was called \emph{dynamic availability}. It was changed here as dynamic availability is already used in literature(cf\cite{lewispye2024permissionlessconsensus})}. 
Resource variability means that the amount of the resource dedicated to mining changes over time. 

Bitcoin can be shown to be secure under resource variability as long as the honest parties hold more hashing power than a potential adversary at any point in time.

With PoStake the situation is more interesting. A PoStake based longest-chain protocol using the Bitcoin chain-selection rule where one picks the ``heaviest'' chain is not secure due to bootstrapping attacks.\footnote{More precisely, assume there's a point in time where a very high amount of coins is staked, say at the $i$th block the honest parties staked $c^h_i$ coins, while the adversary $\adv$ controls a $1/\spratio<1$ fraction of that, i.e., $c^a_i=c^h_i/\spratio$. At this point, $\adv$ bootstraps a private fork $b_i\hookleftarrow b_{i+1}  \ldots \hookleftarrow b_{i+T}$ of some length $T$, each block having weight $c^a_i$. If for the next $T$ blocks the amount staked by the honest parties is (on average) sufficiently smaller than $c^a_i$, the chain created by the honest parties will look lighter than the private fork of $\adv$ at block $i+T$, and at this point $\adv$ can release his private fork which will be adapted by all honest parties.
} On the positive side, the paper on Ouroboros Genesis~\cite{BGKRZ2018} shows that a completely different chain selection rule does imply security even for PoStake based chains (with some additional assumptions, like assuming honest parties delete old keys). Their chain selection rule basically says that given two chains $A$ and $B$ one only looks at the weight of a short subchain starting at the point where $A$ and $B$ fork, and picks the chain whose subchain is heavier. 

For PoSpace based chains the genesis chain selection rule is \emph{not} secure, in fact, unlike the heaviest chain rule, the genesis rule is insecure even without resource variability (i.e. when we assume the space of the honest and adversarial parties is static). There is a simple attack exploiting bootstrapping and \emph{replotting}, which we'll sketch below.
Informally, the reason this attack does not apply in the PoStake or PoW settings is that there's no analog of replotting in PoStake, while in the PoW setting, we do not have bootstrapping.

\subsection{Modelling a Longest-Chain PoSpace Blockchain}
To model a PoSpace based longest-chain blockchain we will make a few idealizing assumptions. As our main result is a lower bound, this only makes our result stronger, concretely
\begin{description}
\item[Resource:] We assume the chain grows by exactly one fresh block per time unit, and each block exactly reflects the amount of space that was used. In reality, a blockchain like Chia or Bitcoin (in the PoW setting) only approximately reflects this amount. One can get a very good approximation of the space used by looking at a sufficiently long subsequence (this idea is used when recalibrating the difficulty). Alternatively one could consider a blockchain design where each block contains the best $k$ proofs for some $k>1$. The larger $k$, the lower the variance and thus the better the approximation. With ``block'' we do not necessarily model a single block, but rather the appearance of a fresh challenge, and this challenge can be used for multiple blocks (e.g. in Chia we have a fresh challenge every 10 minutes, but as Chia uses the correlated randomness technique to prevent double dipping, this challenge is used for up to $64$ actual blocks).
\item[Attacks:]While we model bootstrapping and replotting, we assume there is no grinding or double dipping. This is justified as we have techniques to mostly prevent griding and double dipping, while there's no simple way to prevent replotting, and to prevent bootstrapping we need additional primitives like VDFs.
\item[Resource Variability:] The adversary can control the change in resource, but is restricted to change it only within some $\dynratio$ factor with each block, where $\varepsilon>0$. The quantitative lower bound and the matching upper bound depend crucially on this parameter.
\end{description}

\subsection{Approach for Lower Bound}
To prove our lower bound we let an adversary specify two possible forks, $A$ and $B$ of a chain by specifying how the space of the honest parties changes over time in both cases. 
Now assume we show that (for given ranges of parameters) by exploiting bootstrapping and replotting it is possible to create such forks where $B$ can be ``faked'' using a fraction (say half) of the space the honest parties had in $A$, and vice versa, i.e., $A$ can be faked using half the space of $B$.

This means that an adversary in a hypothetical world where $A$ is the honest chain could fake chain $B$ and vice versa, thus, no matter which chain selection rule is used, in one of the two worlds the adversary's fork will succeed (say the chain rule prefers $A$ over $B$, then in a ``world'' where $B$ is the correct profile, the adversary can create a fork $A$ which will win over $B$).

For our upper bound, we show that a particular chain selection rule is secure almost up to the parameters for which our lower bound applies.

\begin{figure}[h]
    \centering
     \includegraphics[width=0.5\textwidth]{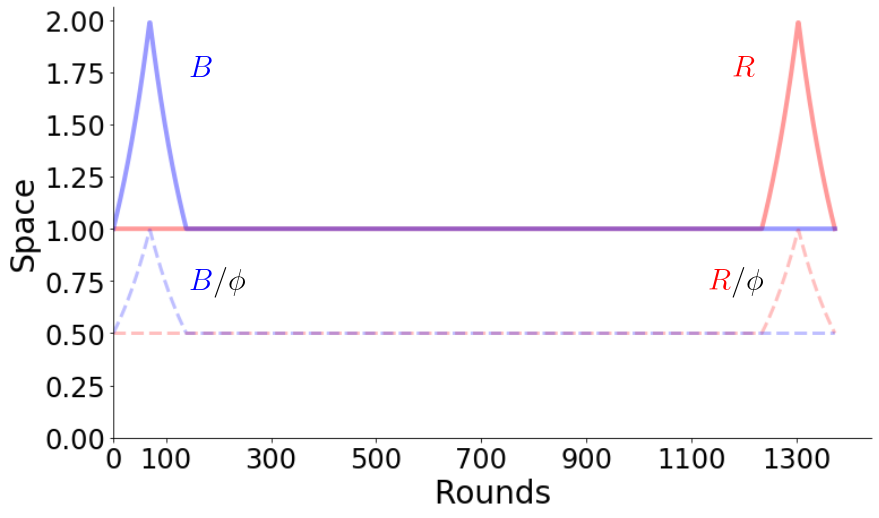}
    \caption{Two profiles as used in our lower bound for $\varepsilon=0.01,\spratio=2$ and $\rho=4$. 
    %The blue profile starts at $s^0_h=1$ and then increases for $50$ steps by the maximum allowed factor $\dynratio=1.02$ to reach $s_h^{50}=(\dynratio)^{50} \approx 2.69$, then drops again down to $1$, and continues flat for $950$ steps. The adversarial space (i.e., a $\spratio$ fraction of it) is shown in light blue. The red profile is similar, except that the ``tent'' is at the end of the profile.
    }
    \label{fig:profiles}
\end{figure}

\subsection{Proof Sketch for Lower Bound}

To prove our lower bound, we specify two profiles ${\color{red}R}$ and ${\color{blue}B}$ reflecting the space honest users have, and then outline a strategy for how (by using bootstrapping and replotting) one can create profile ${\color{red}R}$ using a profile ${\color{blue}B}/\spratio$, where only a $1/\spratio$ fraction of the space in ${\color{blue}B}$ is available (i.e., the space the adversary has in world $A$) and vice versa. Moreover, in the profiles ${\color{blue}B}$ and ${\color{red}R}$ the space is only allowed to change by a $\dynratio$ per block. We will sketch the main idea using the profiles in~\cref{fig:profiles}. The profile ${\color{red}R}$ is plotted by the light red solid line and profile ${\color{blue}B}$ is plotted by the light blue solid line. The solid lines show how much space honest parties control and the dashed lines of the respective colors show how much adversary controls for the respective solid space profile. The honest parties start at space $1$, then we let the profile ${\color{blue}B}$ increase to $\spratio$ as fast as allowed (i.e., by a factor $\dynratio$ per step), then we go back to $1$ as fast as possible, and then there is a long flat part (the length will depend on our parameters). Profile ${\color{red}R}$ is the mirrored version of ${\color{blue}B}$.

\label{sec:proofSketch}
Let us now sketch how the ${\color{red}R}$ (shown by light red solid line in~\cref{fig:profiles}) profile can be faked using a $1/\spratio$ fraction of ${\color{blue}B}$ (shown by solid light blue in~\cref{fig:profiles}). This is illustrated in the figures in~\cref{fig:2a,fig:2b}. The adversary does nothing until block $70$ when its space ${\color{blue}B}/\spratio$ reaches its maximum.

At this point ${\color{blue}B}/\spratio\ge 1$ and the adversary can bootstrap the flat part of ${\color{red}R}$ for $1233$ steps as shown in the top left graph of~\cref{fig:create}. At this point, the adversary only needs to fake the ``tent'' in the last $140$ steps of the ${\color{red}{R}}$ profile. 
For this, it uses replotting.  

%Honest parties start with some space $\hsp_0$ while the adversarial space $\asp_0$ is some $\spratio$ fraction of that. 
%Now the adversary can specify a ``resource profile'' by specifying the number of steps, and for each step, by how much the honest (and in parallel also its own) space can change, this change must be within a factor of $\dynratio$, let us summarize those variables and parameters below

%\begin{figure}[h]
%    \centering
%    \includegraphics[width=6cm]{images/BluetoRedAttackPartial.png}
%\includegraphics[width=6cm]{images/BluetoRedAttack.png}
%    \includegraphics[width=6cm]{images/RedtoBlueAttackPartial.png}
% \includegraphics[width=6cm]{images/RedtoBlueAttackPartial2.png}    
%    \caption{The two figures on top outline how the red profile from Fig~\ref{fig:profiles} is faked using the blue profile for parameters $\varepsilon=0.01,\spratio=2$ and $\rho=4$. In the first step, we use bootstrapping to create the flat part of the red profile (once the blue profile reaches the ``peak'', a $1/\spratio$ fraction of the blue profile is as high as the flat part of the red profile, and thus it can be bootstrapped). Then we use replottnig to create the ``tent'' of the red profile (as $\rho=4$, i.e., replotting takes four steps, it is sufficient that the remaining area below the blue profile is as large as the area under the red ``tent''. 
%    The two figures below illustrate how the blue profile is faked using the red one.}
%    \label{fig:create}
%\end{figure}

\begin{figure}[h]

\centering

\begin{subfigure}[b]{0.5\textwidth}
    \centering
    \includegraphics[width=\textwidth]{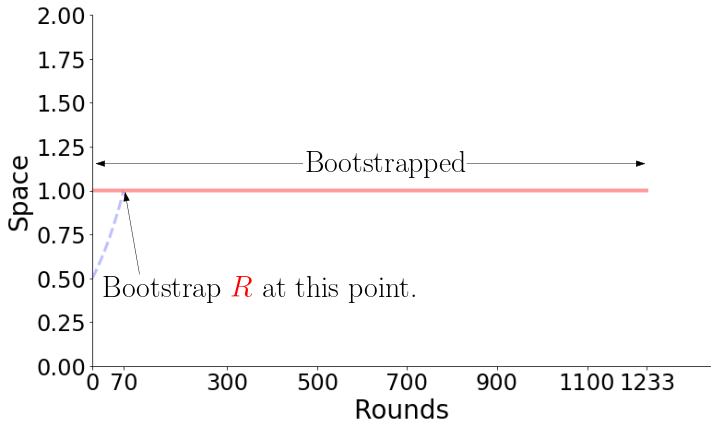}
    \caption{}
    \label{fig:2a}
\end{subfigure}\hfill
\begin{subfigure}[b]{0.5\textwidth}
    \centering
    \includegraphics[width=\textwidth]{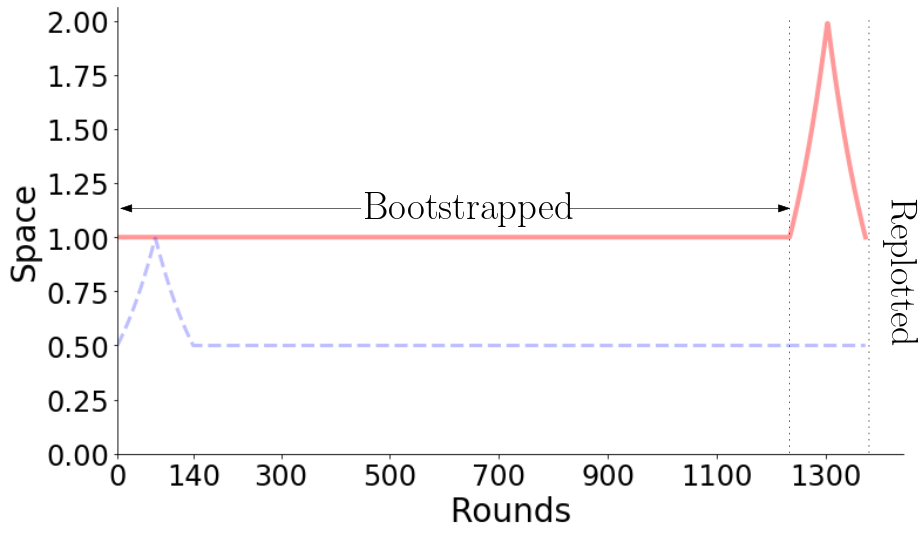}
    \caption{}
    \label{fig:2b}
\end{subfigure}\hfill
\begin{subfigure}[b]{0.5\textwidth}
    \centering
    \includegraphics[width=\textwidth]{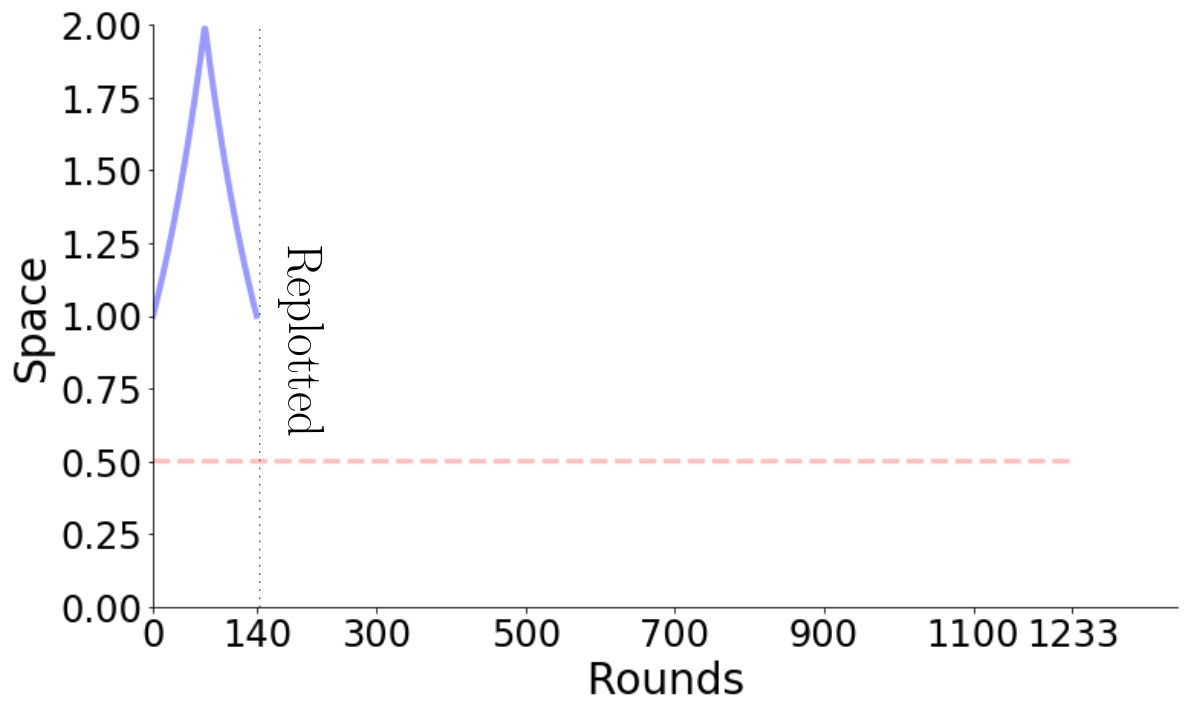}
    \caption{}
    \label{fig:2c}
\end{subfigure}\hfill
\begin{subfigure}[b]{0.5\textwidth}
    \centering
    \includegraphics[width=\textwidth]{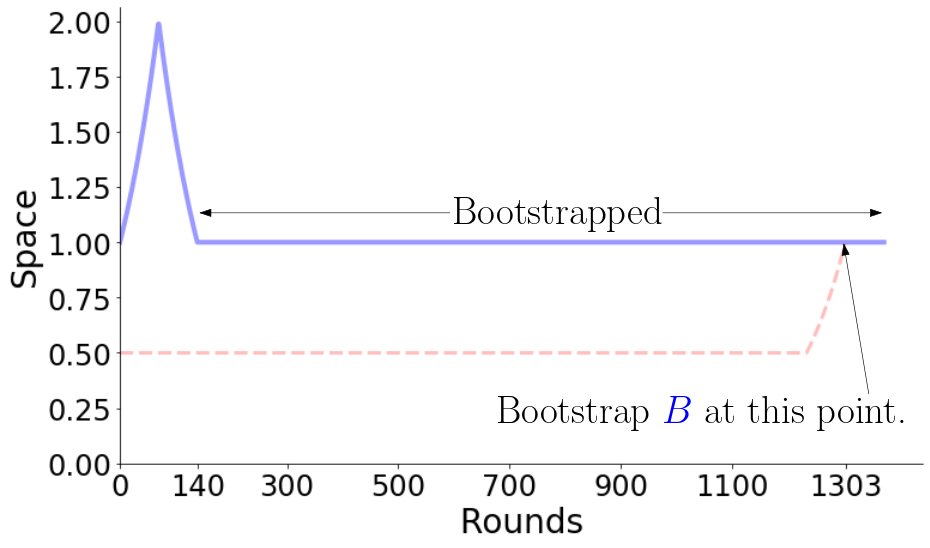}
    \caption{}
    \label{fig:2d}
\end{subfigure}
\caption{The figs (a), (b) outline how the red profile from Fig~\ref{fig:profiles} is faked using the blue profile for parameters $\varepsilon=0.01,\spratio=2$ and $\rho=4$. In the first step, we use bootstrapping to create the flat part of the solid red profile (once the solid blue profile reaches the ``peak'', a $1/\spratio$ fraction of the solid blue profile is as high as the flat part of the solid red profile, and thus it can be bootstrapped). Then we use replotting to create the ``tent'' of the solid red profile (as $\rho=4$, i.e., replotting takes four steps, it is sufficient that the remaining area below the blue profile is as large as the area under the red ``tent''. 
    The two figures (b), and (c) illustrate how the blue profile is faked using the red one.}
    \label{fig:create}
\end{figure}

\iffalse
\section{Images}
\begin{figure}
    \centering
    \includegraphics[width=0.5\linewidth]{images/ResourceProfile.png}
    \caption{Resource Profile}
    \label{fig:enter-label}
\end{figure}

\fi
%\input{2_Preliminaries}
\section{Model}
\subsection{Basic Notation for Chains}
In the abstract model of the chain, we assume that time progresses in discrete steps $t_0,t_1,t_2,\ldots$. 
During the $i$th step $(t_{i-1},t_i]$ the honest parties add a (super)block $b_i$. 
We'll always indicate the position of a block in the chain with a subscript like here.

The space available to the honest parties and the adversary at time $i$ is denoted with $\hsp_i$ and $\asp_i$, respectively. 
We assume that each block $b_i$ perfectly reflects the amount of space that was used to create it, which is denoted by $s(b_i)$ (for the genesis block $b_0$ we set $s(b_0)=1$).

We denote with $b_i\la b_{i+1}$ that $b_{i+1}$ is extending (i.e., contains a hash of) block $b_i$. 
%We use the notation $$b_i^j\eqdef b_i\la b_{i+1}\la \ldots \la b_j$$ \\
For a chain $$\C_0^\ell = b_0 \la \ldots \la b_\ell,  $$ we denote with $\sprof{\C_0^\ell}$ the \emph{space-profile} of chain $\C_0^\ell$ which is defined as $(s(b_i))_{i=0}^{\ell}$, the sequence of space used to create each block $b_i$. 
We'll often use the notation $\C_i^j= b_i \la \ldots \la b_j$ to denote subchains.

\begin{mdframed}[linewidth=1pt]
\label{glossary}
    \textbf{Glossary:}
    \begin{itemize}
        \item[] $\hsp_i\in \R$: The space available to the honest parties at step $i$.
        \item[] $\asp_i\in \R$: The space available to the adversary $\adv$ at step $i$.
        \item[] $\spratio>1$: The honest/adversarial space ratio $\forall i: \hsp_i/\asp_i=\spratio$
        \item[] $\varepsilon > 0$: The rate of change of the space $\forall i: \hsp_i\cdot\frac{1}{\dynratio}\le  \hsp_{i+1} \le \hsp_i\cdot(\dynratio)$
        \item[] $\rho\in\mathbb{N}^+_{\ge 2}$: Number of steps required for replotting
    \end{itemize}
\end{mdframed}

\subsection{Chain Selection Rules}
\label{sec:csr}
A chain selection rule takes as input two chains of the same length and outputs  bit indicating the ``winner''
$$
\CS\colon{\cal C}\times{\cal C}\rightarrow \{0,1\}
$$
Consider two chains $\C_0^\ell,\tilde\C_0^\ell$ 
$$
\C_0^\ell = b_0\la \ldots\la b_\ell\qquad 
\tilde\C_0^\ell=\tilde b_0\la \ldots \la \tilde b_\ell
$$
The Bitcoin chain selection rule simply picks the chain of higher weight, adapted to our space setting this ``highest weight rule'' $\CS_{w}$ is
$$
\CS_{w}(\C_0^\ell,\tilde\C_0^{\ell'})=0 \iff \sum_{i=0}^\ell s(b_i) > 
\sum_{i=0}^{\ell'} s(\tilde b_i) 
$$ Note that the two chains may not always be of the same length but for our lower bound result we can assume without loss of generality that they are equal. 

Most proposed and deployed longest-chain blockchains use a highest weight rule like this, in some cases augmented with checkpointing or finality gadgets that prevent miners from replacing their current chain with another chain that forks too far in the past even if it has a higher weight. 
An interesting exception is the rule suggested in Ouroboros Genesis~\cite{BGKRZ2018} which, for some parameter $k\in\mathbb{N}$, chooses the winning chain based only on the weight of the $k$ blocks following the forking point of the two chains. We'll denote the forking point (i.e., the index of the first blocks that differs) by $\lambda$ as it looks like a forking chain, let 
$$\lambda\eqdef \min\{i\ :\ b_i\neq \tilde b_i\}$$
The genesis chain selection rule, adapted to our setting, can now be defined as
$$
\CS^k_{\mathrm{genesis}}
(\C_0^j,\tilde\C_0^{j})=0 \iff \sum_{i=\lambda+1}^{\min(j,\lambda+k)} s(b_i) > 
\sum_{i=\lambda+1}^{\min(j,\lambda+k)} s(\tilde b_i) 
$$
The Ouroboros genesis rule was introduced as a proof-of-stake based longest-chain blockchain that is secure under resource variability. In~\cref{lem:genesis} we observe that for a proof-of-space based chain, this rule is not secure, and our impossibility result from~\cref{thm:main} shows that in fact, no secure rule exists.

\subsection{Adversarial Options}
%\paragraph{Farming without the Adversary.}
We consider an experiment where honest farmers create a chain (the ``honest chain'') following the rules, while an adversary tries to create a private fork that can at some point be released and will be chosen as the winner over the honest chain by the chain selection rule.

Before specifying the game let us describe the option the adversary has in this game. Concerning the honest chain, if the adversary doesn't contribute at all, the honest parties at the end of 
the $j$th step (i.e., time $t_j$) will have created and agreed on a chain
$$
b_0\la b_1\la \ldots b_j\textrm{ where }\forall i\in[j]\ :\ s(b_i)=\hsp_i
$$
The adversary could contribute to the honest chain, which would create a chain where
$$
b_0\la b_1\la \ldots b_j\textrm{ where }\forall i\in[j]\ :\ \hsp_i\le s(b_i)\le \hsp_i+\asp_i
$$
Intuitively, for the chain selection rule, there shouldn't be any advantage for an adversary to contribute to the honest mining as it should only make the honest chain better, and this will be the case in our attack proving the lower bound. 

While the honest parties will always create the $i$th block in the chain at time $i$, the adversary $\adv$ who creates his private fork must not adhere to this, his only constraint is that his fork must have the same length as the honest chain when it is released. 

In a PoSpace based chain, there are two ways in which $\adv$ can exploit this, bootstrapping and replotting, outlined below.

\paragraph{Bootstrapping.}If at time $i$ $\adv$ knows of a (prefix of a) chain $\C_0^j=b_0\la\ldots\la b_j$, they can extend it immediately to a longer chain. The only constraint is that the space profile of the new blocks is at most the space available to $\adv$, i.e., the new chain satisfies
$$\C_0^j\la b_{j+1}\la\ldots\la b_k\textrm{ where }\forall i>j\ :\ s(b_i)\le \asp_i
$$
\paragraph{Replotting.}
The replotting parameter $\replottime\in\mathbb{N}^+$ specifies how many time steps it takes $\adv$ to replot its space. By replotting the space 
$k$ times -- which requires $k\cdot \replottime$ steps -- they can create a superblock that looks as if they had $k+1$ times the space they actually hold. Formally, using replotting, at time $i$, they can start extending a chain $\C_0^j$ with an extra block $\C_0^j\la b_{j+1}$ where $s(b_{j+1})\le k\cdot \asp_i$, and this will be done by time $i+k\cdot \replottime$.\footnote{In an actual chain, the parameter $\rho=T_{replot}/T_{block}$ is defined by $T_{replot}$, the clock time required to replot, divided by $T_{block}$, the clock time in-between challenges (which is simply the block arrival time if the challenge for a block depends on the previous block like in Bitcoin). In practice, $T_{block}$ should be large enough for a message to spread across the network, which is a few seconds. How large $T_{replot}$ is, depends on many things, most importantly, on the type of PoSpace used. In the PoSpace based on function inversion~\cite{AbusalahACKPR17}  initialization is parallelizable, and thus $T_{replot}$ can be in the order of seconds if the attacker has enough compute power (in particular, GPUs). As a consequence, $\rho$ can only be assumed to be a small constant. In pebbling-based PoSpace~\cite{DFKK2015}, initialization is inherently sequential, and  $T_{replot}$ (and thus $\rho$) is much larger.}

\subsection{The Forking Game}
%\ahad{If possible, change the forking game to more modular definitions similar to the other paper.}
We now define the game in which an adversary $\adv$ forks the honest chain with the goal of fooling the chain selection rule $\CS$ to accept their fork as the winner.

\begin{remark}[Probabilitsic vs. Deterministic]
When analyzing actual blockchains there's always some probabilistic argument, e.g., in Bitcoin the probability that an adversary controlling $X\%$ (for $X<50$) of the hashing power can double spend decreases exponentially with the confirmation time, but it is technically never $0$. Our ``game'' on the other hand is completely deterministic (for given parameters and fork length $\adv$ can win with probability $0$ or $1$) because we assume that each block perfectly reflects the amount of the resource (i.e., space) that was available to create it. While our analysis can be adapted to a probabilistic setting, we don't do so as it does not lead to any more interesting insight. 
\end{remark}
The game is parameterized by $\varepsilon > 1$, controlling how fast the amount of honest space changes; the changes are controlled by the adversary but allowed to change only by a factor $(\dynratio)$ per step. 
\footnote{Assuming that an adversary can precisely control the change of honest space is a strong assumption. But for our lower bound (i.e., an attack for any chain selection rule), arguably natural space profiles suffice. In particular, our attack works for any profile where the honest space stays below some bound $s$ for a longer period of time, with the exception of a ``peak'' of height $\phi \cdot s$ in the middle (the more ``narrow'' this peak is, the shorter the overall period where the profile is below $s$ can be). To break particular chain selection rules, even less demanding profiles are enough, e.g. for the rule used in Spacemint, it's sufficient that the profile at some point in time starts to decrease sufficiently much.}
The parameter $\spratio>1$ controls the amount of honest vs. adversarial space while $\replottime\in\mathbb{N}^+$ is the number of steps required for replotting.

The $(\spratio,\varepsilon,\replottime,\CS)$-game is defined as follows:
\begin{enumerate}
 %   \item $\adv$ chooses the amount of its initial space $\asp_0$.
    \item
    \begin{itemize}
        \item The round counter is set to $i:=0$ and the ``replotting lock'' $lock:=0$.
        \item The honest and adversarial chains are initialized 
    with the genesis block $\C_0^0 = \tilde \C_0^0=b_0$ (where w.l.o.g. $s(b_0)=1 $).
    \item $\adv$ chooses its initial space $\asp_0$ and we set the honest space to $\hsp_0:=\asp_0\cdot \spratio$.
    \end{itemize}
    \item In each round
    \begin{itemize}
        \item Increase the round counter $i:=i+1$.
        \item  %If $lock=0$ (i.e., no replotting going on), 
        $\adv$ chooses the space adjustment 
        $\gamma_i$ in the range $\frac{1}{(\dynratio)}\le \gamma_i\le  (\dynratio)$
        and the space is set to 
        $$
        \asp_i:=\asp_{i-1}\cdot \gamma_i\quad,\quad \hsp_i:=\asp_i\cdot \phi
        $$
        \item The honest chain is extended $$\C_0^i:=\C_0^{i-1}\la b_i$$ with a block with space profile $s(b_i)=\hsp_i$    
        \item
        If $lock>0$ (i.e., replotting is going on) set $lock:=lock-1$, otherwise $\adv$ can extend its current chain $\tilde \C_0^j$ in two ways
        \begin{description}
        \item[bootstrap: ]Extend $\tilde \C_0^j$ to 
        $$\tilde \C_0^{j'}=\C_0^j\la \tilde b_{j+1}\la\ldots\la \tilde b_{j'}$$ where 
        $\forall t,j+1\le t\le j' \ :\ s(\tilde b_t)\le \asp_i$.
        \item[replot: ]$\adv$ can call a replot request by which the last block 
        $\tilde b_j$ is replaced with a block $\tilde b'_j$ with space profile
        $$
        s(\tilde b'_j)\le s(\tilde b_j)+\asp_i  
        $$
        Set the replotting lock to $lock:=\replottime-1$.
        \end{description}
        \item if $lock=0$ (i.e., no replotting going on) and the length $j$ of the adversarial chain $\tilde B_0^j$ is at least $j\ge i$, then $\adv$ can stop the game. \\

        \end{itemize}

    If the chain selection rule prefers the current ($i$ block prefix of the) adversarial chain to the honest one, i.e.,
        $$
        \CS(\C_0^i,\tilde \C_0^i)=1
        $$
        then we say the game is $\ell$-winning for the $\adv$, where $\ell$ denotes the length of the fork (i.e., length of chain minus the length of the common prefix)
        $$\ell = i- \max\{k\ :\ b_k=\tilde b_k\}.$$
\end{enumerate}
\subsection{Forking Existing Rules}
In this section, we'll observe that the forking game can be won against the highest weight $\CS_{w}$ and the genesis $\CS^k_{\mathrm{genesis}}$ chain selection rules that we discussed in~\cref{sec:csr}. To break $\CS_w$ one only requires bootstrapping (but no replotting) and resource variability, i.e. a $\varepsilon > 0$. For $\CS_{\mathrm{genesis}}^k$ the $\spratio$ can be $1$.
\begin{lemma}
    The $(\spratio,\varepsilon,\replottime,\CS_{w})$-game can be $\ell$-won for $\ell=\lint{\frac{\spratio}{\varepsilon}}$
  \end{lemma}
\begin{proof}
    To win $(\spratio,\varepsilon,\replottime,\CS_{w})$-game, $\adv$ simply bootstraps a long chain and then reduces the amount of space dynamically in order to make the honest chain have weight less than the adversarial chain. $\adv$ never contributes to the honest chain. 
    
    Concretely, in each step $i > 0$ the adversary decreases the space by the maximum allowed amount $\hsp_{i+1} = \frac{1}{(\dynratio)}\cdot \hsp_{i}$. At $i=1$ $\adv$ bootstraps
    \[
        \tilde\C_0^{j} = \C_0^0 \la \tilde b_{1}\ \la \tilde b_{2}\ \la \cdots \la \tilde b_{j}\ 
    \] where $ s( \tilde b_{t}\ ) = \hsp_0/\spratio= 1/\spratio$ for all $t \in [1, j]$. After this adversary simply lets $\C_0$ catch up. The game ends on round $j$. This gives,
     $$\textit{Weight of } \C_0^j = \sum_{t=0}^{j} \left( \frac{1}{(\dynratio)} \right )^{t} = \frac{1 - \frac{1}{(\dynratio)^{j+1}}}{1 - \frac{1}{(\dynratio)}} < \frac{1}{1-\frac{1}{(\dynratio)}}$$
     while $$\textit{Weight of } \tilde \C_0^j\ = 1 + \frac{j}{\spratio}$$ 

     Thus when $j\geq \lint{\frac{1}{(\dynratio)- 1}\cdot \spratio}$ the weight of the adversarial chain is higher than the weight of the honest chain. Hence $\CS_{w}(\C_0^j, \tilde \C_0^j\ ) = 1$.\qed
\end{proof}

\begin{lemma}
\label{lem:genesis}
    The $(\spratio, \varepsilon, \replottime, \CS_{\mathrm{genesis}}^k)$-game can be $\ell$-won for $\ell = \lint{\spratio} \cdot k \cdot \replottime$ 
\end{lemma}
\begin{proof}
    To win $(\spratio, \varepsilon, \replottime, \CS_{\mathrm{genesis}}^k)$-game, $\adv$ simply replots many times to get an adversarial chain which in its $k$ blocks after the fork has enough weight to be larger than the weight of the honest chain in the corresponding $k$ blocks. 
    
    Concretely, the game runs till $l = \lint{\spratio}\cdot k \cdot \replottime$. Throughout the game, $\adv$ does not use the resource variability; $\hsp_i, \asp_i$ remain constant at $1$, $\frac{1}{\spratio}$ respectively. Further, $\adv$ doesn't contribute to the honest chain. This produces $$ \C_0^j = b_0 \la b_1 \la \cdots \la b_l$$ such that $s(b_i) = 1 \forall i \in [0,l]$.  

    $\adv$ does the following:
    \begin{enumerate}
        \item On round $i=1$, it forks the chain to create $$\til \C^1_0 \ = b_0 \la \til b_1 $$ where 
        $s(\til b_1) = \frac{1}{(\dynratio)}$. Set $j:=1$
        \item Starting with round $i=1$ and ending on round $l$, it does the following steps:
        \begin{enumerate}
            \item If $(i-1) \mod \lint{\spratio} \cdot \replottime = 0$ and $i>1$, put $j := j+1$ and add a new block $\til b_j\ $, with $s(\til b_j) = \frac{1}{\spratio}$ to the adversarial chain $$ \til \C_0^j := b_0 \la \til b_1\ \la \cdots \la \til b_{j-1}\ \la \til b_j. $$ Then $\adv$ puts $lock := \replottime -1$ and starts replotting on $\til b_j$. 

            \item For next $\lint{\spratio}\cdot\replottime - 1$ rounds $\adv$ repeats replotting on $\til b_j$ to increase its space to $\lint{\spratio}/\spratio \ge 1$. Go back to step $(a)$.
        \end{enumerate}

        \item In the last round $\adv$ bootstraps the chain to create additional $l - k$ blocks to get a chain of length $l+1$.
    \end{enumerate}

    The adversarial chain now is $$\til \C_0^l\ = b_0 \la \til b_1\ \la \cdots \la \til b_{l-1}\ \la \til b_l.$$ Here $\lambda = 0$ as the chains $\C_0^l, \til\C_0^l\ $ forked at the first block. Thus we get $\sum_{i=1}^{\min(l,k)} s(b_i) = k$ while $\sum_{i=1}^{\min(l,k)} s(\til b_i) \ge k$. Hence $\CS_{\mathrm{genesis}}^k(\C_0^j, \til \C_0^j\ ) = 1$ and $\adv$ wins.
    \qed   
\end{proof}

\subsection{Our Contribution}
Having introduced the forking game, we can now state our main result that under resource variability, no PoSpace longest-chain blockchain is secure. 

\begin{restatable}[Impossibility Result]{theorem}{thmmain}
\label{thm:main}
For every chain selection rule $\CS$, there exists an  adversary $\adv$ that wins the $(\spratio,\varepsilon,\rho, \CS)$-forking game in
\begin{align*}
    \ell &= \lint{\rho\cdot\spratio^2\cdot (\dynratio) \cdot \left( \frac{(\dynratio) - \frac{1}{\spratio}}{\varepsilon} \right)} \\&+ \lint{\rho\cdot\spratio^2\cdot \left( \frac{(\dynratio) - \frac{1}{\spratio}}{\varepsilon}\right)} + 2\cdot \lint{\frac{\log \spratio}{\log(\dynratio)}} \textrm{ steps.}
\end{align*}
\end{restatable}

%\begin{theorem}[Impossibility Result]
%\label{thm:main}
%For every chain selection rule $\CS$, there exists an  adversary $\adv$ that wins the $(\spratio,\varepsilon,\rho, \CS)$-forking game in
% \begin{align*}
%     \ell &= \lint{\rho\cdot\spratio^2\cdot (\dynratio) \cdot \left( \frac{(\dynratio) - \frac{1}{\spratio}}{\varepsilon} \right)} \\&+ \lint{\rho\cdot\spratio^2\cdot \left( \frac{(\dynratio) - \frac{1}{\spratio}}{\varepsilon}\right)} + 2\cdot \lint{\frac{\log \spratio}{\log(\dynratio)}} \textrm{ steps.}
% \end{align*}

% \end{theorem}
%\paragraph{Proof Sketch}
We sketched the general proof idea in~\cref{sec:proofSketch}. The full formal proof can be found in the~\cref{app:1}. While the expression in the theorem is somewhat complicated, typically we'd assume that 
$\varepsilon$ is small, say $<0.1$, while $\phi$ is bounded away from $1$, say $\phi\ge 1.1$. In this case the bound becomes $$\ell \le O(\rho \cdot\spratio^2/\varepsilon)\ .$$

Next, we prove that the fork length in the attack from~\cref{thm:main} is an optimal attack up to a factor $\spratio$. We show this by providing a simple chain selection rule $\CS_\mathrm{tent}$ such that $\adv$ can not win a $(\spratio, \dynratio, \rho, \CS_\mathrm{tent})$ if the fork length is less than $\rho\left(\spratio \cdot (\dynratio)\cdot\frac{1 - \frac{1}{\spratio}}{\varepsilon} - \lint{\frac{\log \spratio}{\log(\dynratio)}} \right)$. 

\sloppy Before we give proof, we'll make a definition that will be useful. A $ \tent = (\spratio, x, y)$-tent, with $x \in \mathbb{N}$, $y\in \R$, is the infinite sequence $$\ldots,y_{x-1},y_{x},y_{x+1},\ldots $$ where for $i > x$, $y_i=y_{i-1}/\spratio$ and for $i < x$, $y_{2x-i}=y_{2x-i+1}/\spratio$ (for an illustration see~\cref{fig:profiles}, where the first $36$ steps of the dark blue curve are part of a $(\spratio=1.02,x=18,y=2)$ tent. We say $y$ is the size of the tent $\Gamma = (\spratio, x, y)$ and that tent $\tent = (\spratio,x, y)$ is larger than tent $\tent' =(\spratio,x', y')$ if $y>y'$.

\begin{restatable}[The attack from~\cref{thm:main} is tight up to $\spratio$]{theorem}{thmlowerBound}
%\begin{theorem}[The attack from~\cref{thm:main} is tight up to $\spratio$]
\label{thm:lowerBound}
For any $(\spratio,\dynratio,\rho)$ there's a chain selection rule $\CS_\mathrm{tent}$  for which no adversary can win the $(\spratio,\dynratio,\rho,\CS_\mathrm{tent})$ forking game in less than 
$$
\rho\left(\spratio \cdot (\dynratio)\cdot\frac{1 - \frac{1}{\spratio}}{\varepsilon} - \lint{\frac{\log \spratio}{\log(\dynratio)}} \right)\textrm{ steps}
$$
%\end{theorem}
\end{restatable}

The full proof can be found in~\cref{app:2}.

\section{Discussion and Open Problem}
In this paper, we showed that there's no chain selection rule that guarantees security (against double spending) for a permissionless longest-chain blockchain based on proofs of space assuming honest parties always hold more space than an adversary.

\paragraph{Overcoming our No-Go Result.} Recall that our attacker can replot space and bootstrap the chain. Two existing PoSpace based chains, Chia and Filecoin, avoid our impossibility in different ways. Chia prevents bootstrapping by additionally using proofs of time, while Filecoin avoids replotting by using a BFT (rather than longest-chain) type protocol as we'll elaborate below.

%using a chain selection rule using only Proof of Space is secure under resource variability, thus other building blocks are required to design a secure blockchain using space as a mining resource.

Chia~\cite{chia} combines proofs of space with ``proofs of time'', where the latter are instantiated with verifiable delay functions~\cite{BBBF2018}. 
One can think of (a simplified version of) Chia as simply alternating PoSpace with VDFs, where the challenge for the next VDF (PoSpace) is computed from the previous PoSpace (VDF output). 
Bootstrapping such a chain is not possible as the main security property of a VDF requires that computing its output requires time.

In Filecoin parties must register their space before it can be used for mining. Moreover, parties must constantly compute and publish proofs for their registered space. Blocks are then created by the parties who registered space in a BFT-type protocol. Informally, this prevents replotting as only registered space can be used, and registering more space than one actually controls is not possible as one must constantly prove that space is available. Using the classification from \cite{lewispye2024permissionlessconsensus}, one can see our results as being in the fully permissionless setting (where the protocol has no knowledge about current participation), while Filecoin works in a quasi-permissionless setting (where parties must be known to the protocol and be always available).

A question one can ask is whether simply committing to space without periodic checks would overcome our impossibility result. This would correspond to \emph{dynamic availability} setting in~\cite{lewispye2024permissionlessconsensus}. The answer is no: an adversary can simply plot and commit to many different proofs of space in the honest chain and then later replot them when launching an attack. Thus our result precludes PoSpace based Nakamoto like longest chain blockchain in both fully permissionless and dynamic availability setting. 

\paragraph{Open Problems.}
Our upper and lower bounds are separated by a gap $\phi$, we believe this gap can be closed by coming up with a more sophisticated chain selection rule for Theorem~\ref{thm:lowerBound}. 

\ahad{TODO: Add discussion about how Chia avoids attacks. Provide intuition. Add discussion about how commitments along with some repeated checking, like in filecoin, can avoid attacks. Though naively committing to the chain doesn't avoid it.}

\ahad{Should we also add context around fully permissionless vs quasi-permissionless settings of Lewis-Pye Roughgarden's paper?}

%\vspace{\fill}%
%\pagebreak[0]%
%\vspace{-\fill}%
%\pagebreak
%\input{extra}

\bibliographystyle{splncs04}
\bibliography{bib_camera_ready}
%\pagebreak
\appendix
\section{Proofs}

\subsection{Proof of \cref{thm:main}}\label{app:1}

\thmmain*
\begin{proof}[Proof of~\cref{thm:main}]
We already sketched the general idea in~\cref{sec:proofSketch} and will make it more formal here.

To prove the theorem we must specify two space profiles 
$\mathcal{S} = (\hsp_i)_{i=0}^{\ell}$ and $ \Tilde{\mathcal{S}} = ( \hsptil_i)_{i=0}^\ell$ where $\hsp_0=\hsptil_0$ and 
${\hsp_i}/{(\dynratio)}\le \hsp_{i+1}\le \hsp_i\cdot (\dynratio)$ and ${\hsptil_i}/{(\dynratio)}\le \hsptil_{i+1}\le \hsptil_i\cdot (\dynratio)$(same for $\hsptil$) such that 
\begin{itemize}
    \item Using $\asp_i = \hsp_i/\spratio$ the adversary can create a chain $$\til \C_0^l\ = b_0 \hookleftarrow \til b_1\ \hookleftarrow \cdots \hookleftarrow \til b_{\ell-1}\ \hookleftarrow \til b_\ell.$$ with space profile $ (\hsptil_0,\hsptil_1,\ldots, \hsptil_\ell)$

    \item Using $\asptil_i = \hsptil_i/\spratio$ the adversary can create a chain $$\C_0^l\ = b_0 \hookleftarrow b_1\ \hookleftarrow \cdots \hookleftarrow b_{\ell-1}\ \hookleftarrow b_\ell.$$ with space profile $(\hsp_0,\hsp_1,\ldots, \hsp_\ell)$
\end{itemize}

For this we first define $k \coloneqq \lint{\frac{\log \spratio}{\log(\dynratio)}}$ and $l \coloneqq \lint{\rho\cdot\spratio^2\cdot (\dynratio) \cdot \left( \frac{\dynratio - \frac{1}{\spratio}}{\varepsilon} \right)} + \lint{\rho\cdot\spratio^2\cdot \left( \frac{\dynratio - \frac{1}{\spratio}}{\varepsilon}\right)}$. Notice, $\ell = 2k+l$. 

In the first space profile in the first $k$ rounds space increases by a factor $\dynratio$ each round and then in the next $k$ rounds it decreases by a factor $\dynratio$ each round. In the remaining $l+1$ rounds the space stays constant at $1$. In the second space profile the roles of $k$ and $l+1$ rounds are reversed; in the first $l+1$ rounds the space stays constant at $1$. Then in the next $k$ rounds space increases by a factor $\dynratio$ each round and finally in the last $k$ rounds it decreases by a factor $\dynratio$ each round.      Formally,
\[
    \hsp_i = \begin{cases} 
                (\dynratio)^i &\text{for } 0\le i \le k-1\\
                (\dynratio)^{2k-i} &\text{for } k \le i \leq 2k-1\\
                1 &\text{for } 2k \leq i \leq l+2k 
            \end{cases}
\]
and 
\[
    \hsptil_i = \begin{cases}
        1 &\text{for } 0 \leq i \leq l\\
        (\dynratio)^{i-l-1} &\text{for } l+1 \le i \le l+k\\
        (\dynratio)^{l+2k-i} &\text{for } l+k+1 \le i \leq l+2k
    \end{cases}
\]

\begin{lemma}[$\adv$ creates $\til {\cal S}$ from $\cal S$]
\label{lem:Sprof1toSprof2}
    An adversary, $\adv$, playing $(\spratio, \dynratio, \rho, \CS)$-forking game can create a chain $\til \C_0^\ell$ such that $\sprof{\til \C_0^\ell} = ({\hsptil_i})_{i=0}^{\ell}$ while honest chain is $\C_0^\ell$ with $\sprof{\C_0^\ell} = ({\hsp_i})_{i=0}^{\ell}$ 
\end{lemma}
\begin{proof}[Proof of~\cref{lem:Sprof1toSprof2}]
    The adversary, $\adv$, does following to create $\til \C_{i=0}^\ell$ while honest chain is $\C_{i=0}^\ell$:
    \begin{enumerate}
        \item At round $i=0$, $$\C_0^0 = \til \C_0^0 = b_0$$ where $s(b_0) = 1$.
        \item For $1 \le i \le k-1$, set $\dynratio_i = \dynratio$. Thus $\hsp_i = \dynratio\cdot\hsp_{i-1}$. Honest chain becomes $$\C_0^i = \C_0^{i-1} \hookleftarrow b_i $$ where $s(b_i) = (\dynratio)^i$. The adversarial chain $$\til \C_0^0 = b_0$$ remains unchanged. 
        \item For $i = k$, set $\dynratio_i =\dynratio$. So, $\hsp_i = \dynratio\cdot\hsp_{i-1} = (\dynratio)^k \geq \spratio$. Thus $\asp_i \geq 1$. Honest chain is $$\C_0^i = \C_0^{i-1} \hookleftarrow b_i$$ where $s(b_i) = (\dynratio)^i$. Now $\adv$ bootstraps the adversarial chain to become $$\til \C_0^l = b_0 \hookleftarrow \til b_1 \hookleftarrow \ldots \hookleftarrow \til b_l$$ where $s(\til b_j) = 1$ for all $j \in [1, l]$. This is demonstrated in~\cref{fig:2a}. 
        \item For $k+1 \le i \le 2k$, set $\dynratio_i = \frac{1}{(\dynratio)}$. Thus, $\hsp_i = (\dynratio)^{2k-i}$. Honest chain becomes $$\C_0^i = \C_0^{i-1} \hookleftarrow b_i$$ where $s(b_i) = (\dynratio)^{2k-i}$. The adversarial chain remains at $\til \C_0^l$. 
        \item For $2k+1 \le i \le l+2k$, set $\dynratio_i = 1$. Thus $\hsp_i = 1$. Honest chain continues as $$\C_0^{i} = \C_0^{i-1} \hookleftarrow b_i$$ where $s(b_i) = 1$. For the adversarial chain, $\adv$ uses replotting to create $$\til \C_{l+1}^{l+2kl} = \til b_{l+1} \hookleftarrow \ldots \hookleftarrow \til b_{l+2k}$$ such that $$s(b_{l+i}) = \begin{cases}
            (\dynratio)^i &\text{for } 1\le i \le k\\
            (\dynratio)^{2k-i} &\text{for } k+1 \le i \le 2k.
        \end{cases}$$
        This is demonstrated in~\cref{fig:2b}.
        Thus it achieves an adversarial chain $$\til \C_0^\ell = \til \C_0^{l} \hookleftarrow \til \C_{l+1}^{l+2k}$$ with the space profile $\til {\cal S}$. \\

        \label{para:replottent}To see why replotting can achieve the requisite space profile, note that $\asp_i = 1/\spratio \: \forall i \in [l+1, l+2k]$. In order to create a block $b$ such that $s(b) = \alpha$, $\adv$ needs to replot $\lint{\frac{\alpha - \frac{1}{\spratio}}{\frac{1}{\spratio}}} = \lint{\alpha\cdot\spratio - 1} $ times and this would take $\rho \cdot \lint{\alpha\cdot\spratio - 1}$ rounds. Thus the total number of rounds required is \begin{align*}
            &\sum_{i=1}^k\rho\cdot\lint{(\dynratio)^i\cdot \spratio - 1} + \sum_{i=k+1}^{2k}\rho\cdot\lint{(\dynratio)^{2k-i} \cdot \spratio - 1} &\\ &\le \rho\cdot\sum_{i=1}^k\lint{(\dynratio)^i\cdot \spratio - 1} + \rho\cdot\sum_{i=0}^{k-1}\lint{(\dynratio)^{i} \cdot \spratio - 1}&\\
            &\le \rho\cdot\sum_{i=1}^k(\dynratio)^i\cdot \spratio + \rho\cdot\sum_{i=0}^{k-1}(\dynratio)^{i} \cdot \spratio&\\
            &= \rho \cdot \spratio \cdot (\dynratio)\cdot \frac{(\dynratio)^{k} - 1}{\varepsilon} + \rho\cdot \spratio \cdot \frac{(\dynratio)^{k} - 1}{\varepsilon} &\\
            &\le \rho \cdot \spratio \cdot (\dynratio)\cdot \frac{\spratio \cdot (\dynratio) - 1}{\varepsilon} + \rho\cdot \spratio \cdot \frac{\spratio \cdot (\dynratio) - 1}{\varepsilon} \\&\left(\text{as } \spratio \le \lint{\spratio} \le (\dynratio)^k \le \spratio \cdot (\dynratio) \right) \\
            &= \rho \cdot \spratio^2 \cdot (\dynratio) \cdot \frac{\dynratio - \frac{1}{\spratio}}{\varepsilon} + \rho\cdot \spratio^2 \cdot \frac{\dynratio - \frac{1}{\spratio}}{\varepsilon} &\\
            &\le \lint{\rho\cdot\spratio^2\cdot (\dynratio) \cdot \left( \frac{\dynratio - \frac{1}{\spratio}}{\varepsilon} \right)} + \lint{\rho\cdot\spratio^2\cdot \left( \frac{\dynratio - \frac{1}{\spratio}}{\varepsilon}\right)} = l
        \end{align*} 
        Since total number of rounds is $\ell = l + 2k$, the adversary $\adv$ after round $2k$ has $l$ rounds to replot, which is sufficient.  
    \end{enumerate}
    At round $i = l+2k = \ell$ we have the honest chain as $\C_0^\ell$ with space-profile $\cal S$ and adversarial chain as $\til \C_0^\ell$ with space-profile $\til {\cal S}$. This completes the proof of~\cref{lem:Sprof1toSprof2}. 
    \qed
\end{proof}

\begin{lemma}[$\adv$ creates ${\cal S}$ from $\til {\cal S}$]
\label{lem:Sprof2toSprof1}
    An adversary, $\adv$, playing $(\spratio, \dynratio, \rho, \CS)$-forking game can create a chain $\C_0^\ell$ such that $\sprof{\C_0^\ell} = ({\hsp_i})_{i=0}^{\ell}$ while honest chain is $\til \C_0^\ell$ with $\sprof{\til \C_0^\ell} = ({\hsptil_i})_{i=0}^{\ell}$ 
\end{lemma}
\begin{proof}[Proof of~\cref{lem:Sprof2toSprof1}]
    To create $\C_{i=0}^\ell$ while honest chain is $\til \C_{i=0}^\ell$, the adversary, $\adv$, does the reverse of~\cref{lem:Sprof1toSprof2}, $i.e.$ it first replots and then bootstraps. Formally,
    \begin{enumerate}
        \item At round $i=0$, $$\C_0^0 = \til \C_0^0 = b_0$$ where $s(b_0) = 1$. Note that here $\til \C$ is the honest chain while $\C$ is the adversarial chain. 
        \item For round $1\le i \le l$, set $\dynratio_i = 1$. Thus, $\hsp_i = 1$. The honest chain becomes $$\til \C_0^i = \til \C_0^{i-1} \hookleftarrow \til b_i$$ where $s(\til b_i) = 1$. $\adv$ replots a chain $$\C_0^{2k} = b_0 \hookleftarrow b_1 \hookleftarrow \ldots \hookleftarrow b_{2k} $$ such that $$s(b_{i}) = \begin{cases}
            (\dynratio)^i &\text{for } 1\le i \le k\\
            (\dynratio)^{2k-i} &\text{for } k+1 \le i \le 2k.
        \end{cases}$$
        This is demonstrated in~\cref{fig:2c}. 
        As argued in~\cref{lem:Sprof1toSprof2}~\cref{para:replottent} $l$ rounds are sufficient to replot $\C_0^{l}$. 
        \item For round $l+1 \le i \le l+k-1$, set $\dynratio_i = \dynratio$. Thus $\hsp_i = (\dynratio)^{i-l}$. The honest chain becomes $$\til \C_0^{i} = \til \C_0^{i-1} \hookleftarrow \til b_i$$ where $s(\til b_i) = (\dynratio)^{i-l}$. The adversarial chain remains unchanged. 
        \item For round $i=l+k$, set $\dynratio_i = \dynratio$. Thus $\hsp_i = (\dynratio)^k \ge \lint{\spratio}$. Therefore, $\asp_i \ge \frac{\lint{\spratio}}{\spratio} \ge 1$. The honest chain continues as $$\til \C_0^i = \C_0^{i-1} \hookleftarrow \til b_i$$ where $s(\til b_i) = (\dynratio)^k$. $\adv$ uses space $\asp_i \ge 1$ to bootstrap the adversarial chain to form $$\C_0^{2k+l} = \C_0^{2k} \hookleftarrow b_{2k+i} \hookleftarrow \ldots \hookleftarrow b_{2k+l}$$ where $s(b_{2k+i}) = 1 \: \forall i \in [1, l]$. This is demonstrated in~\cref{fig:2d}.  
        \item For round $l+k+1 \le i \le l+2k$, set $\dynratio_i = \frac{1}{(\dynratio)}$. Thus $\hsp_i = (\dynratio)^{l+2k-i}$. The honest chain continues as $$\til \C_0^{i} = \til \C_0^{i-1} \hookleftarrow \til b_i$$ where $\displaystyle s(\til b_i) = (\dynratio)^{l+2k-i}$. The adversarial chain remains unchanged.
    \end{enumerate}
    At round $i = l+2k = \ell$ we have the honest chain as $\til \C_0^\ell$ with space-profile $\cal {\til S}$ and adversarial chain as $\C_0^\ell$ with space-profile $ {\cal S}$. This completes the proof of~\cref{lem:Sprof2toSprof1}.
    \qed
\end{proof}

Consider any chain selection rule $\CS$. From~\cref{lem:Sprof1toSprof2,lem:Sprof2toSprof1} we get two forking games: first, where the honest chain is $\C_0^{\ell}$ and $\adv$ creates $\til \C_0^\ell$ and second, where the honest chain is $\til \C_0^\ell$ and $\adv$ creates $\C_0^\ell$. If $\CS(\C_0^\ell, \til \C_0^\ell) = 0$, then in the first forking game the chain selection rule would choose the adversarial chain. If $\CS(\C_0^\ell, \til \C_0^\ell) = 1$, then in the second forking game the chain selection rule would choose the adversarial chain. Therefore in either case it is possible to fool the chain selection. Note that we made no restriction on $\CS$; we only used replotting and resource variability. Hence, we can conclude that there does not exist a secure chain selection rule under resource variability. 
\qed
\end{proof}

\subsection{Proof of \cref{thm:lowerBound}}
\label{app:2}

\thmlowerBound*

\begin{proof}
Given two chains $\C_0^j=b_0\hla \cdots \hla b_j$ and 
$\til C_0^j=\til b_0\hla \cdots\hla \til b_j$
selection rule first determines the forking point $f$, i.e., the first block where chains diverge to have $b_f\neq \til b_f$. Let $l \coloneqq j-f$ denote the fork length.

\sloppy Let ${\cal S} =(a_1,\ldots,a_l)\eqdef (s(b_f),\ldots,s(b_j))$ and $\til {\cal S}=(\til a_1,\ldots,\til a_l)\eqdef (s(\til b_f),\ldots,s(\til b_j))$ be the space profiles of the fork. The $\CS_{\mathrm{tent}}$ rule now picks the chain whose fork covers the larger tent. More formally, let $\mu$ be maximal such that there exists a tent $\Gamma = (\spratio,x,\mu)$ under $\cal S$, where $f \le x \le j$ and for all $i,0\le i\le j$ we have $a_i \ge y_i$. Similarly, we define tent $\til \Gamma = (\spratio, \til x, \til \mu)$ for $\til {\cal S}$. With these definitions, the chain selection rule is defined as $\CS_\mathrm{tent}(\C_0^j,\til \C_0^j)=0$ if $\mu\ge \til \mu$.

Let's assume $\C_0^j$ is the honest chain, while $\til C_0^j$ is the adversarial one created in the $(\spratio,\dynratio,\rho,\CS_\mathrm{tent})$ forking game. To win the game it must hold that 
$\til \mu>\mu$. We know that during the fork (when the honest parties created $b_f,\ldots,b_j$) the honest parties never controlled more $\mu$ space (otherwise we could have embedded a tent larger than $\mu$ under $A$), and thus the adversary never controlled more than $\mu/\phi$ space in that phase. But during that phase he must have created a chain with space profile $\til {\cal S}$, which in particular contains a tent of size $\til \mu$, and thus also a tent $\til \Gamma' = (\spratio, \til x, \mu)$ of size $\mu$ (as $\mu < \til \mu$). While the adversary might have used bootstrapping to create this it could only have bootstrapped space up to $\mu/\spratio$. 

We want a lower bound on the amount of space above $\mu/\spratio$ in the tent $\til \Gamma'$ after the forking point. Let $\til z \coloneqq \til x - f$ and $k = \lint{\frac{\log \spratio}{\log(\dynratio)}}$. We have the following scenarios 

\begin{enumerate}
    \item If $\til z \ge k$ or $l - \til z \ge k$, then the fork is long enough to contain a sequence of $\mu, \frac{\mu}{(\dynratio)}, \cdots, \frac{\mu}{(\dynratio)^{k-1}}$ above the $\frac{\mu}{\spratio}$. Number of rounds required to replot this space, using $\frac{\mu}{\spratio}$, is 
    \begin{align*}
    \sum_{i=0}^{k-1}\rho\cdot\lint{\frac{\frac{\mu}{(\dynratio)^i} - \frac{\mu}{\spratio}}{\frac{\mu}{\spratio}}} &= \sum_{i=0}^{k-1}\rho\cdot\lint{\frac{\spratio}{(\dynratio)^i} - 1} \\
    \ge \rho \cdot \sum_{i=0}^{k-1}\left(\frac{\spratio}{(\dynratio)^i} - 1\right) &= \rho \cdot \left( \spratio \cdot (\dynratio) \cdot \frac{1 - \frac{1}{\spratio}}{\varepsilon} - k\right).
    \end{align*}

    \item If $\til z < k$ and $l-\til z < k$, then the fork is too short for replotting as every step of the tent in the forked chain is above $\mu/\spratio$ and for each of them we need at least $\rho$ steps of replotting. Thus in total at least $\rho*l$ steps are required. Since $\rho \ge 2$, replotting is not possible. Therefore, $\adv$ cannot win the game and we have a contradiction.     
\end{enumerate}

This finishes the proof that to win $(\spratio,\dynratio,\rho,\CS_\mathrm{tent})$ forking game the fork length must be at least $\rho \cdot \left( \spratio \cdot (\dynratio) \cdot \frac{1 - \frac{1}{\spratio}}{\varepsilon} - k\right)$.
\qed
\end{proof}

\end{document}